\documentclass{article}

\usepackage{booktabs}

\usepackage{times}
\usepackage{enumitem}
\usepackage{rotating}
\usepackage{array}
\usepackage{mathtools}
\usepackage{bm}
\usepackage{breqn}
\usepackage{comment}
\usepackage{enumitem}
\usepackage{graphics}
\usepackage{multirow}
\usepackage{tikz}
\usepackage{todonotes}
\usepackage{caption}
\usepackage{subcaption}
\usepackage{boldline}
\usepackage{color}
\usepackage{xcolor}
\usepackage{soul}
\usepackage[utf8]{inputenc}
\usepackage{amsfonts}
\usepackage{amsthm}
\setlist[itemize]{leftmargin=*}

\usepackage{array,supertabular}
\usepackage{amsmath,amssymb,amsthm}
\usepackage{tikz}
\usepackage{relsize}
\usepackage{enumerate}

\usepackage{todonotes}

\usepackage{caption}
\usepackage{float}

\usepackage{bookmark}
\usepackage{longtable}

\usepackage{tablefootnote}
\usepackage{multirow}
\usepackage{verbatim,graphics,graphicx,mathrsfs}
\usepackage{enumitem}
\usepackage{balance}
\usepackage[linesnumbered,ruled,vlined]{algorithm2e}
\SetKwProg{Fn}{Function}{}{}
\usepackage{float}
\usepackage{multirow}
\usepackage{tablefootnote}
\pgfdeclarelayer{background}
\pgfsetlayers{background,main}
\usepackage{tikz}
\usetikzlibrary{arrows,positioning,shapes,fit,calc,decorations,plotmarks,shapes,backgrounds}

\newcommand{\set}[1]{\{#1\}} 
\newcommand{\midd}{\mathrel{:}}

\newcommand{\commentout}[1]{}
 \newlength{\wordlength}

\makeatletter
\newcommand\footnoteref[1]{\protected@xdef\@thefnmark{\ref{#1}}\@footnotemark}

\newtheorem{theorem}{Theorem}

\usepackage{wrapfig}
\usepackage{multirow}
\usepackage{verbatim,graphics,graphicx,mathrsfs}
\usepackage{enumitem}
\usepackage{balance}
\usepackage{amsmath,amsfonts}
\usepackage{amssymb}
\usepackage{listings}
\usepackage{wrapfig}
\usepackage{tablefootnote}

\newcommand{\MC}{\mathit{MC}}

\usepackage{footnote}

\usepackage{soul}
\usepackage{graphicx}
\usepackage{tabularx}
\usepackage{mathrsfs}
\usepackage{longtable}
\usepackage{multirow}
\usepackage{amsmath}
\usepackage{amsfonts}
\usepackage{bbm}
\usepackage{etoolbox}
\usepackage{booktabs}
\usepackage{enumitem}
\usepackage{url}
\usepackage{algpseudocode}

\usepackage{threeparttable}

\def\MC{\mathit{MC}}

\usepackage{bm}

\newcommand{\Nod}{\mathit{Nod}}
\newcommand{\dist}{\mathit{dist}}
\newcommand{\dists}{\mathit{dists}}

\begin{document}
\title{A Game-Theoretic Algorithm for Link Prediction}  
\author{Mateusz Tarkowski, University of Oxford\\
Tomasz P. Michalak, University of Warsaw\\
Michael Wooldridge, University of Oxford}

\maketitle

\begin{abstract}
Predicting edges in networks is a key problem in social network analysis and involves reasoning about the relationships between nodes based on the structural properties of a network. In particular, link prediction can be used to analyse how a network will develop or---given incomplete information about relationships---to discover ``missing'' links. Our approach to this problem is rooted in cooperative game theory, where we propose a new, quasi-local approach (i.e., one which considers nodes within some radius $k$) that combines generalised group closeness centrality and semivalue interaction indices. We develop fast algorithms for computing our measure and evaluate it on a number of real-world networks, where it outperforms a selection of other state-of-the-art methods from the literature. Importantly, choosing the optimal radius $k$ for quasi-local methods is difficult, and there is no assurance that the choice is optimal. Additionally, when compared to other quasi-local methods, ours achieves very good results even when given a suboptimal radius $k$ as a parameter.
\end{abstract}

\section{Introduction}
In this paper we are concerned with link prediction---an important research problem in social network analysis \cite{barbieri2014,liben2007,Lu2011}. The aim is to predict between which pairs of nodes unknown links exist or should form based on the known structural characteristics of the network. Applications include biological networks, where some of the network is known, but checking whether other links actually exists is very costly (e.g., food web, protein-protein and metabolic networks \cite{Lu2011}); social networks, where certain information about links is impossible to attain or verify (e.g., covert networks \cite{waniek2019hide}); and identifying potentially fruitful collaboration in organisations \cite{liben2007}.

Predicting whether links exist or will form is strongly associated with the concept of node similarity \cite{Lu2011}. If two nodes are similar, then there is a greater chance that they are connected. The three most common approaches to computing node similarity take into account either \textit{local}, \textit{quasi-local} or \textit{global} topological information about nodes. In general, global methods produce better results than local ones, but are more computationally involved. Furthermore, quasi-local methods, which only consider the network within a certain radius $k$ around the pair of nodes, tend to do better than global ones \cite{Lu2011} when given an optimal radius. This is because, it is unlikely that features of the network that are far away from a pair of nodes (that are taken into account by global methods) will actually impact the chance of those two nodes being connected.

Recently, Szczepanski et al. \cite{Szczepanski:2015b} proposed a new method to tackle the link prediction problem based on group $k$-degree group centrality and the concept of the game-theoretic interaction index. In particular, $k$-degree group centrality is a valuation of groups of nodes according to the number of nodes that are at a distance $k$ or closer to the group \cite{Michalak:et:al:2013}. In turn, the interaction index can be interpreted as a measure of the synergy \cite{alshebli2019measure} or similarity of players in a cooperative situation. When applied to $k$-degree group centrality, the interaction index becomes an interesting measure of the similarity of the topological placement of nodes. In particular, two nodes produce negative synergy according to $k$-degree group centrality whenever they have common neighbours. The interaction index is used to measure such negative synergies between
any pair of nodes in the context of all the possible groups of nodes that these two nodes belong to. Considering all these groups allows the measure to prioritise the impact of unique common neighbours. The intuition behind considering all groups is as follows: if most of the nodes in a network neighbour a certain node, then this does not necessarily mean that they are similar; on the other hand, if only two nodes neighbour a certain node, then this is a unique characteristic shared only by these two nodes and suggests that they are similar. 
Szczepanski et al. show that their game-theoretic approach produces a very competitive link prediction method when compared to a selection of quasi-local measures.  As for computational considerations, the authors develop a general algorithm for computing interaction indices of $k$-degree centrality that runs in $O(|V|^4)$ and a more specific one that runs in $O(|V|^3)$ time. 

In what follows, we propose an alternative game-theoretic method to tackle the link prediction problem. Our method builds upon the following intuition: while the $k$-degree centrality used by Szczepanski et al.~\cite{Szczepanski:2015b} postulates that all nodes within a distance $k$ of a pair of nodes impact their similarity \textit{equally}, we postulate that \textit{those that are farther away should be considered less important than those that are closer from the pair}. In order to develop a method based on this intuition, we use the concept of generalised group closeness centrality, which has precisely this property: \textit{the magnitude of the impact of a node on the value of a group reduces as the distance from this group increases}.\footnote{See the work by Skibski and Sosnowska \cite{skibski2018axioms} for an overview of different centrality measures that focus on the distances between nodes.} As compared to $k$-degree centrality, this allows us to reduce the impact of far-away nodes on synergy. As we are interested in developing a quasi-local measure, similarly to $k$-degree centrality, we introduce a restriction on the generalised group closeness centrality that only nodes inside of a radius $k$ around a given group impact the value of this group value.

Interestingly, despite the fact that our method is more general that the one developed by Szczepanski et al.~\cite{Szczepanski:2015b}, our algorithm for computing it is less demanding computationally. In particular, we compute our generalised closeness semivalue interaction index with a complexity of $O(V_k^2|V|^2)$, where $V_k$ is the average number of nodes within a distance $k$ of any node. We also develop an algorithm to compute the generalised closeness Shapley value interaction index that runs in $O(V_k^2|V| + |V|^2)$ time. This is an interesting result, since computing the interaction index given an arbitrary cooperative game is difficult ($\#$-P complete, like the Shapley value \cite{Chalkiadakis:et:al:2011,Deng:Papdimitriou94}). We should also highlight the generality of our algorithm stemming from the fact that we allow for any function of distance, $f$, to be used with our measure. For $f(d) = 1$ whenever $d < k$, our measure is equivalent to the $k$-degree interaction index.

We evaluate our approach by considering $11$ real-life networks on which we run a series of experiments. We find that---in most cases---our approach produces equal or better results than the state-of-the-art quasi-local measures in the literature. Among others, it outperforms the game-theoretic method proposed by Szczepanski et al.~\cite{Szczepanski:2015b}. 

Furthermore, we consider an aspect of quasi-local measures that has thus-far been overlooked in the literature. In particular, most authors present the result of their measure given an optimal parameter $k$ \cite{Szczepanski:2016b,Szczepanski:2015b,Lu2011}. However, the optimal value of this parameter is unknown and it must be estimated. Therefore, we study how the measures perform given a selection of values for the parameter $k$, and not only the optimal value of this parameter, since there is always a chance that the estimated value of $k$ is not optimal.

Two cases can be distinguished. If a value for $k$ is chosen that is smaller than the optimal value of this parameter, then this results in a computational/accuracy trade-off. The measure can usually be computed faster, but it could perform better by having more information about the network (i.e., a larger value for $k$). For large networks, it may simply take too long to compute a quasi-local measure with an optimal $k$, and such a trade-off is necessary. In the worst case, $k=1$ may be chosen, and the quasi-local measure becomes local. On the other extreme, if $k$ is equal to the longest shortest path in the network, then it can be said that the quasi-local measure becomes global.

What happens, then, if  the value of $k$ is larger than the optimal value? In our study, we show that such a $k$ can result in a significant reduction of the quality of the similarity ranking for quasi-local measures. We find, however, that our proposed measure is largely resistant to this. For example, in one network that we study---a Football network \cite{Girvan:2002,pajek}---the precision of the ranking produced by the measure due to Szczepanski et al.~\cite{Szczepanski:2015b} decreases by approximately $37\%$ when the value of $k$ is $3$ instead of the optimal value $1$, whereas the precision of our measure decreases by less than $1\%$. Interestingly, we also note that there are cases where the quality of other measures is reduced given a larger $k$, while the quality of our measure improves. This means that whereas other measures pick up more noise given more information about the network, our measure is still able to gather more insight in order to improve the ranking of potential edges.

\section{Preliminaries}
In this section we introduce the key concepts required for the understanding of the paper.

\subsection{Graph-Theoretic Concepts}
A \emph{network} is a directed \emph{weighted graph}~$G=(V, E,\omega)$, where~$V$ is a set of nodes,~$E$ is a set of edges, i.e., unordered pairs~$(v,u)$ of nodes in~$V$ with $v\neq u$, and $\omega: E \rightarrow \mathbb{R}^+$ is a weight function from edges to the positive real numbers. A graph is \textit{unweighted} if $\omega(e)=1$ for all $e\in E$.
We denote the \emph{neighbours} of a node~$v$ by $E(v)=\set{u\midd (v,u)\in E}$ and the neighbours of a subset~$C$ of nodes by $E(C)=\bigcup_{v\in C}E(v)\setminus C$.
We refer to the \emph{degree} of a node~$v$ by $\mathit{deg}(v) = |E(v)|$. We define the distance from a node~$s$ to a node~$t$ as the length of the shortest path from $s$ to $t$ and denote it by~$\mathit{dist}(s, t)$, and we define the distance between a node~$v$ and a subset of nodes~$C \subseteq V$ as $\mathit{dist}(C, v) = \min_{u \in C} \textit{dist}(u, v)$. 

The \textit{Generalised Group Closeness Centrality} \cite{Freeman:1979,Everett:Borgatti:1999,Michalak:et:al:2013b} of a group of nodes~$S$ in a graph $G$ is defined as:
\begin{equation}
\nu^{CL}_{f}(G)(S) = \sum_{v \in V \setminus S} f(dist(S,v));
\label{equation:general_group_closeness}
\end{equation}
\subsection{Game-Theoretic Concepts}
A cooperative game is defined by a group of players, $I$, and a characteristic function $\nu: 2^I \rightarrow \mathbb{R}$ that assigns to each group of players a real value, with the restriction that $\nu(\emptyset) = 0$. For our purposes, the players are nodes within  a graph (i.e., $I = V$). A central concept to cooperative game theory is that of the marginal contribution, i.e., how much value a player $i$ brings to a coalition $C$ of players. Formally, $MC_\nu(C, i) = \nu(C \cup \{i\}) - \nu(C)$. 
Whereas in the game-theoretic literature marginal contributions are used to measure the importance of players for a group in order to divide the value of said group fairly among its members~\cite{Chalkiadakis:et:al:2011}, we are concerned with using them in order to rank the similarity of players. To achieve this,
We follow Owen~\cite{Owen:1972} in defining the synergy between two players, $i$ and $j$, within the context of a coalition $C$ as the difference between the marginal contribution of the group $\{i,j\}$ and the marginal contributions of each node separately. Formally,
\begin{align}
&S_\nu(C,i,j) = \nonumber\\
&\MC_\nu(C, \{i,j\}) - \MC_\nu(C, \{i\}) - \MC_\nu(C, \{j\}).
\label{equation:S}
\end{align}
The Shapley value interaction index of $i$ and $j$  is defined as:
\begin{align}
I^{\mathit{Shapley}}_{i,j}(\nu) = \sum_{\pi \in \Pi(I^{i \wedge j}) } \frac{S_\nu(C_\pi(\{i,j\}),i,j)}{(n - 1)!},
\label{equation:Shapley_interaction_index}
\end{align}
where $I^{i \wedge j} = I \setminus \{i,j\} \cup \{\{i,j\}\}$, $\Pi(X)$ is the set of permutations of the set $X$, and $C_\pi(\{i,j\})$ is the set of elements preceding $\{i,j\}$ in the permutation $\pi$. 
Grabisch and Roubens~\cite{Grabisch:1999} continued this work and introduced the Banzhaf interaction index. Szczepanski et al.~\cite{Szczepanski:2015b} generalised these concepts further by introducing semivalue interaction indices, defined below:
\begin{equation}
I^{\mathit{SEMI}}_{i,j}(\nu) = \sum_{k = 0}^{n-2} \sum_{C \in C^k(I \setminus \{i,j\})} \beta(k) \frac{S_\nu(C,i,j)}{\binom{n - 2}{k}}
\label{equation:semivalue_interaction_index}
\end{equation}
For our purposes, the lower the interaction index of two nodes, the more similar they are.

\subsection{Performance Metrics for Link Prediction}
Next, we present the main metrics used to evaluate link prediction methods. The first of these is the \textit{area under curve} (AUC), and the second is precision \cite{Lu2011}.
The goal of a link prediction algorithm is to identify pairs of nodes that are not connected by an edge but which do (or will) exist in the real world (i.e., nodes that are ``missing'' from the graph). For example, people in an online social network are typically connected by their friendships. However, this is not to say that people who have not indicated their friendship via the social platform are not actually friends. Ideally, a link prediction algorithm would identify those pairs of individuals who are friends in real-life but not (yet) on the social network.

In order to test link prediction methods, it is typical to remove a certain percentage of links from a real life network \cite{Lu2011,Szczepanski:2015b,Szczepanski:2016b}. This altered network is given as input to a link-prediction algorithm, which ranks disconnected pairs of nodes in terms of the likelihood that they belong to the ``removed,'' or ``missing'' set.

\vspace{0.2cm}\noindent \textbf{$\blacktriangleright$ Area Under Curve}
The area under the ROC curve is a good overall indicator of the quality of the ranking of edges that is produced by a link prediction algorithm and we can compute it using the Mann-Whitney $U$ test \cite{hanley:1982}. Formally, let $m$ be the number of edges in the network that are ``missing.'' In other words, these edges do not exist in the graph that the link-prediction algorithm received as input, but do exist in the real world. These are the edges that we would like a link-prediction algorithm to discover and therefore give them a high rank. Let $l$ be the number of edges that do not exist in the graph that are not missing. In other words, these are the edges that we would like a link-prediction algorithm to recognise as non-existent and therefore give them a low rank.
There are a total of $n = m \times l$ comparisons between missing and non-existent edges. Let $n'$ be the number of such comparisons where a missing edge is ranked over a non-existent edge, and $n''$ be the number of comparisons where a missing edge is given the same rank as a non-existent edge. AUC, then, is defined as follows: \[\mathit{AUC} = \frac{n' + \frac{n''}{2}}{n}\]
If all missing edges are ranked higher than non-existent edges, the resulting AUC is equal to $1$. If the opposite is true (i.e., all non-existent edges are ranked higher than all missing edges), then this results in an AUC of $0$. A random link-prediction algorithm results in an average AUC of $0.5$.
In short, AUC is the percentage of comparisons between ranked edges that are ``correct.'' For this reason, we present AUC as a percentage value.


\vspace{0.2cm}\noindent \textbf{$\blacktriangleright$ Precision}
For a given $p \in \mathbb{N}$, let $\mathit{Top}(p)$ be the set of top $p$ edges according to the ranking generated by a link-prediction algorithm. Let $\mathit{Correct}(p)$ be the set of edges in $\mathit{Top}(p)$ that are ``missing.'' The precision of the algorithm, then, is defined as follows: \[\mathit{Precision}(p) = \frac{|\mathit{Correct}(p)|}{p}\]
If all of the top $p$ edges ranked by an algorithm are missing edges, then the precision is equal to $1$. If none of them are, then the precision is equal to $0$. Note that this metric is dependent on the variable $p$, which indicates the depth to which the ranking is studied. In particular, this depth should never be higher than the actual number of missing edges.
Whereas most literature \cite{Lu2011} has focused on a static depth $p$, such as $100$, we take a different approach. Since the sizes of our networks vary greatly, a static depth makes no sense, which is why we take $p$ to be equal to the number of missing edges in the network.

\section{Semivalue Closeness Interaction Index \& Its Computation}
\label{section:semivalue:interaction:index}

In this section, we introduce our family of quasi-local measures for link prediction. Generally, quasi-local measures are characterised by requiring a parameter, $k$. When a quasi-local algorithm evaluates the likelihood of an edge existing between the nodes $u$ and $v$, it only ever considers the nodes that are at distance of at most $k$ away from $u$ or $v$. If all other nodes that are farther than $k$ from $u$ and $v$ would be removed from the graph, then this would not change the evaluation of the existence of an edge between $u$ and $v$ according to a quasi-local algorithm.
Conversely, local algorithms are equivalent to quasi-local algorithms with $k = 1$ and global algorithms consider the whole graph in evaluating the likelihood of an edge existing.
Following Szczepanski et al.~\cite{Szczepanski:2015b}, we apply semivalue interaction indices (see Equation \ref{equation:semivalue_interaction_index}) to a group centrality measure as a means to measure the similarity of disconnected nodes. Whereas Szczepanski et al.~\cite{Szczepanski:2015b} used group $k$-degree centrality, we use a broader class of group centrality measures---general group closeness centrality, $\nu^{CL}_{f}$: \[\nu^{CL}_{f}(G)(S) = \sum_{v \in V} f(dist(S,v)).\]  
If for any natural number $k$ we define $f$ as
\[f(d) = \begin{cases}
1 & \text{ if } d \leq k \\
0 & \text{ otherwise, }
\end{cases}\] then $\nu^{CL}_{f}$ is equivalent to $k$-degree centrality. To develop our measure, however, we use the following distance function:
\[f(d) = \begin{cases}
\frac{1}{d^2} & \text{ if } d \leq k \\
0 & \text{ otherwise. }
\end{cases}\] This choice has two benefits:
\begin{itemize}
\item[(1)] It retains the computational advantage of $k$-degree group centrality, whereby faraway nodes do not impact it. This not only improves computational performance, but also the accuracy of the resulting similarity measure, since Szczepanski et al.~\cite{Szczepanski:2015b,Szczepanski:2016b} and Szczepanski et al.~\cite{Szczepanski:2016b} showed that faraway nodes are less likely to impact the similarity of nodes and therefore reduce the accuracy of the measure. In fact, this is well-known for various similarity measures, which is why some quasi-local measures outperform global ones in many networks \cite{Lu2011}. Faraway nodes do not impact the index, and therefore need not be considered, which leads to faster computation. We show that this also improves the AUC and precision of the index, since faraway nodes are less likely to impact the similarity of nodes. This is why some quasi-local measures outperform global ones \cite{Lu2011}.
\item[(2)] It is likely that nodes that are closer impact similarity more than those that are further away, so it makes sense to use a decreasing function such as $\frac{1}{d^2}$ for those nodes where $d < k$ in order to decrease the impact of far-away nodes. We also studied functions such as $\frac{1}{d}$ or $\frac{1}{2^d}$, but found that $\frac{1}{d^2}$ produced the best results.
\end{itemize}

Let us now introduce our main computational results. We start off by proving that the generalised closeness semivalue interaction index can be computed in polynomial time of $O(V_k^2|V|^2)$. Next, we consider a particular case of this general result, i.e., the generalised closeness Shapley value interaction index. We prove that it can be computed even faster, in $O(V_k^2|V| + |V|^2)$ time. Importantly, for both algorithms, we leave the choice of $f$ open, meaning that the analysis and algorithms can be used with any decreasing function.

\begin{theorem}
$I_{s,t}^{\mathit{Semi}}(\nu^{CL}_f(G))$ can be computed in $O(V_k^2|V|^2)$ time.
\label{theorem:shemival:interaction}
\end{theorem}

\begin{proof}
Our goal is to compute 
\begin{equation}
I^{\mathit{SEMI}}_{s,t}(\nu^{CL}_f(G)) = \sum_{k = 0}^{|V|-2} \sum_{C \in C^k(V \setminus \{s,t\})} \beta(k) \frac{S(C,s,t)}{\binom{|V| - 2}{k}},
\end{equation}
i.e., Equation \ref{equation:semivalue_interaction_index}, in polynomial time.  We will be counting the number of coalitions for which certain common expressions appear in this sum. By multiplying these expressions by their number of appearances, we will achieve polynomial computation. Let us first look closer at the definition of $\nu^{CL}_f(G)$ in Equation \ref{equation:general_group_closeness}. In particular, the equation itself consists of a sum over nodes $u$. We will only focus on one of these elements at a time, keeping in mind that $\sum_{u \in V} I_{s,t}^{\mathit{SEMI}}(f(\dist(C,u)) = I_{s,t}^{\mathit{SEMI}}(\nu^{CL}_f(G))$. Let us define $\nu_u = f(\dist(C,u))$ for the remainder of our proof and focus on computing $I_{s,t}^{\mathit{SEMI}}(\nu_u)$ for some $s$ and $t$.

Moreover, we will focus on computing the inner sum of Equation \ref{equation:semivalue_interaction_index}---$\sum_{C \in C^k(V \setminus \{s,t\})} \beta(k) \frac{S(C,s,t)}{\binom{|V| - 2}{k}}$---for an arbitrary $k$, and our resulting algorithm will then sum the value of this inner part for all $k$ such that $0 \leq k \leq |V| - 2$. We assume without loss of generality that $\dist(s,u) \leq \dist(t,u)$, meaning that $S(C,s,t) = \MC(C,t)$.  Moreover, only coalitions $C$ such that $\dist(C,u) > \dist(t,u)$ matter (since otherwise $S(C,s,t) = 0$). In effect, we arrive at the following simplification:
\[\sum_{C \in C^k(V \setminus \{s,t\})} - \beta(k) \big( \frac{\nu_u(\{s\})}{\binom{|V| - 2}{k}} - \frac{\nu_u(C)}{\binom{|V| - 2}{k}}\big)\].

We will use the following notation:
\[MC^+(s,t,u,k) = \hspace{-0.4cm} \sum_{C \in \big\{C \midd \substack{ C \in C^k(V \setminus \{s,t\}) \text{ and } \\ \dist(C,u) > \dist(s,u) } \big\}} \beta(k) \frac{\nu_u(\{s\})}{\binom{|V| - 2}{k}}\]
\[MC^-(s,t,u,k) = \hspace{-0.4cm} \sum_{C \in \big\{ C \midd \substack { C \in C^k(V \setminus \{s,t\}) \text{ and } \\ \dist(C,u) > \dist(s,u)  } \big\}  } \beta(k) \frac{\nu_u(C)}{\binom{|V| - 2}{k}}\]

The rest of the proof will focus on computing these values. In order do this, let us introduce the following notation:
\[\Nod_{\sim d}(u) = \{v \midd v \in V \text{ and } c \sim u\},\]
where $\sim$ is one of $<$, $>$, $\leq$, or $\geq$.

\paragraph{$\boldmath{MC^+(s,t,u,k)}: $} Let $d = \dist(t,u)$. Computing $MC^+(s,t,u,k)$ is equivalent to computing the following expression: \[|\{C \midd C \in C^k(V \setminus \{s,t\}) \text{ and } \dist(C,u) > d\}|,\] and multiplying it by $\beta(k) \frac{f(d)}{\binom{|V| - 2}{k}}$. We need to count the number of coalitions $C$ of size $k$ such that $ \dist(C,u) > d$. Counting the number of such coalitions is as simple as counting the number of ways to choose $k$ elements from $Nod_{> d}$. In other words: $\binom{Nod_{> d}}{k}$. This gives us the desired result:
\[
MC^+(s,t,u,k) = \beta(k) \frac{f(d)}{\binom{|V| - 2}{k}} \binom{Nod_{> d}}{k}
\]

\paragraph{$\boldmath{MC^-(s,t,u,k)}: $} Let us define
\[MC^-(s,t,u,k,d) = \hspace{-0.4cm} \sum_{C \in \big\{ C \midd \substack{C \in C^k(V \setminus \{s,t\}) \text{ and } \\ \dist(C,u) = d  } \big\}  } \beta(k) \frac{\nu_u(C)}{\binom{|V| - 2}{k}}.\]
We now have \[MC^-(s,t,u,k) = \sum_{d \in \big\{d \midd \substack{ d \in \dists(u) \text{ and } \\ d < dist(s,u) } \big\}} MC^-(s,t,u,k,d).\] We therefore have to find the number of coalitions of size $k$ such that $\dist(C,u) = d$. In other words, they need to have at least some node at distance $d$ from $u$ and no nodes that are closer. The answer is as follows:
$\binom{Nod_{\geq d}}{k} - \binom{Nod_{> d}}{k}$.
This gives us the desired result:
\[
MC^-(s,t,u,k,d) = \frac{f(d)}{\binom{|V| - 2}{k}} \bigg( \binom{Nod_{\geq d}}{k} - 
\binom{Nod_{> d}}{k} \bigg).
\]
Algorithm \ref{algorithm:semi:links} implements the equations from this proof and computes the semivalue closeness interaction index in the required time.
\end{proof}

As for the generalised closeness Shapley value interaction index, the following result holds.

\begin{theorem}
$I_{s,t}^{\mathit{Shapley}}(\nu^{CL}_f(G))$ can be computed in $O(V_k^2|V| + |V|^2)$ time.
\label{theorem:shapleyval:interaction}
\end{theorem}

\begin{proof}
Our proof will be based on dissecting Equation \ref{equation:Shapley_interaction_index}. First, note that this equation is a sum of multiple expressions over various permutations. To achieve polynomial computation, we will group expressions that are equal to one another and count how many permutations these expressions appear in. Finally, by multiplying the expressions by their respective number of appearances we will achieve polynomial computation. As previously, we will focus on computing $I_{s,t}^{\mathit{Shapley}}(\nu_u)$ for some $s$ and $t$, and the resulting algorithm will be a sum over $u \in V$. Again, assuming that $\dist(s,u) \leq \dist(t,u)$ we have  $S(C,s,t) = - MC(C,\{t\}) = - \big(\nu(C \cup \{t\}) - \nu(C)\big)$.

Continuing, our aim will be to dissect the formula for $I_{s,t}^{\mathit{Shapley}}(\nu_u)$ into smaller, more manageable parts, and to compute those. Note that $MC(C,\{t\}) \neq 0$ if and only if $\dist(t,u) < \dist(C,u)$. In this case $\nu(C \cup \set{t}) = \nu(\{t\})$ is independent of $C$. We refer to this as the left, or \emph{positive}, part of the sum that constitutes the marginal contribution. We refer to $\nu(C)$ as the negative part. As previously, we note that $\nu^{CL}_f(G)$ in Equation \ref{equation:general_group_closeness} consists of a sum over nodes $u$ and define  $\nu_u = f(\dist(C,u))$. We will focus on computing our similarity metric for $\nu_u$, and the final answer will be a sum of $\nu_u$ for all $u \neq s,t$. Next, let us introduce the following notation:
\begin{align*}
MC^+(t,u) = \sum_{\pi \in \big\{\pi \midd \substack{ \pi \in \Pi(V^{s \wedge t}) \text{ and } \\ \dist(t,u) < \dist(\pi_{t},u) } \big\} } \nu_u(\{t\}),
\end{align*}
\begin{align*}
MC^-(t,u) = \sum_{\pi \in \big\{\pi \midd  \substack{ \pi \in \Pi(V^{s \wedge t}) \text{ and } \\ \dist(t,u) < \dist(\pi_{t},u) } \big\} } \nu_u (C_\pi(\{s,t\})),
\end{align*}
where $C_\pi(x)$ is the set of elements in the permutation $\pi$ that precedes $x$, and arrive at the following simplification:
\begin{align*}
I_{s,t}^{\mathit{Shapley}}(\nu_u) = - \frac{(MC^+(t,u) - MC^-(t,u))}{(|V|-1)!}. 
\end{align*}
The remainder of the proof will focus on computing $MC^+(t,u)$ and $MC^-(t,u)$.

\paragraph{$\boldmath{MC^+(t,u):}$} Our goal here is to find the number of permutations $\pi \in \Pi(V^{s \wedge t})$ such that $\dist(t,u) < C_\pi(\{s,t\})$ and multiply this number by $\nu(\{t\})$. Let $d = dist(t,u)$. We can construct all such permutations in the following manner:
\begin{itemize}
\item First, choose $\Nod_{\leq d}(u) - 1$ positions (we have to subtract $1$, since $s$ and $t$ are treated as one node) for all of the nodes in $V^{s \wedge t}$ that are as close to $u$ as $t$ or closer. Out of all of these positions, $t$ has to be first, otherwise $\dist(t,u) < C_\pi(t)$  will not be satisfied. There are $\binom{|V| - 1}{\Nod_{\leq d}(u) - 1}$ ways that we can choose the positions and all of the nodes except the first can be permuted in $(\Nod_{\leq d}(u) - 2)!$ ways.
\item Next, the rest of the nodes are placed in the rest of the positions, which can be permuted in $(|V| - 1 - (\Nod_{\leq d}(u))! - 1)$ ways.
\end{itemize}
When we combine both steps, there are $\binom{|V| - 1}{\Nod_{\leq d}(u) - 1}    (\Nod_{\leq d}(u) - 2)!  (|V| - \Nod_{\leq d}(u))!$ such permutations, which simplifies to:
\[MC^+(t,u) = \frac{(|V| - 1)!}{\Nod_{\leq d}(u) - 1}\]

\paragraph{$\boldmath{MC^-(t,u):}$} Let us introduce the following notation 
\[MC^-(t,u,d) = \sum_{\pi \in  \big\{\pi \midd  \substack{ \pi \in \Pi(V^{s \wedge t}) \text{ and } \\ \dist(C_\pi(\{s,t\}),u) = d  }  \big\} } \nu_u (C_\pi(\{s,t\}))\]
and let $\mathit{dists}(u)$ be the set of distances from any node to $u$. In effect, we have \[MC^-(t,u) = \sum_{d \in \big\{d \midd\substack{  d \in \mathit{dists}(u) \text{ and } \\ \dist(t,u) < \dist(d,u)  } \big\} } MC^-(t,u,d).\] For a given $d$, then, we will focus on computing $MC^-(t,u,d)$. We need to capture all permutations $\pi$ for which the coalition of nodes preceding $\{s,t
\}$ (i.e., $C_\pi(\{s,t\}$) is exactly at distance $d$ from $u$. The requirement can be summarised as follows: there needs to be at least one node $x$ in $\pi$ preceding $\{s,t\}$ such that $\dist(x, u) = d$ and no nodes that are closer than $x$ to $u$ preceding $t$. This can be counted using the inclusion/exclusion principle.
 Counting all permutations such that $dist(C_\pi(\{s,t\}),u) \geq d$ and subtracting those such that $dist(C_\pi(\{s,t\}),u) > d$ provides the answer. We can use the techniques for counting $\MC^+(t,u)$ to arrive at the following:
\[MC^-(t,u,d) = \frac{(|V| - 1)!}{\Nod_{< d}(u) - 1} - \frac{(|V| - 1)!}{\Nod_{\leq d}(u) - 1} \]
Algorithm \ref{algorithm:Shapley:links} implements the equations from this proof and computes the Shapley value closeness interaction index in the required time. This concludes our proof. 

\end{proof}

\SetInd{0.3em}{0.4em}
\RestyleAlgo{ruled}
\begin{algorithm}[t]
\SetAlgoVlined
\LinesNumbered
        \SetKwInOut{Input}{input}
        \SetKwInOut{Output}{output}
        \SetKwProg{Fn}{Function}{ is}{end}
        \Input{Graph $G = (V,E, \omega)$, Closeness function $f: \mathbb{R} \rightarrow \mathbb{R}$, Probability distribution function $\beta: 0, 1, \ldots |V| - 1 \rightarrow \mathbb{R}$, radius $k$}
        \Output{Configuration Semivalue}

$dist[V][V]$\;
\For{$v \in V$}{
\For{$u \in V$}{
$dist[v][u] = \infty\;$
}
  $distance[v] \gets$ empty set\;
  $visited \gets$ empty set\;
  $\phi_v \gets 0$\;
  $Q \gets $ Priority Queue\;
  $Q.enqueue(\langle v,0 \rangle)$\;

  $dist[v][v] = 0$\;
  \While{$Q$ not Empty}{
  	$\langle u,d \rangle \gets Q.pop()$\;
    $II[v,u] = 0$\;
    $visited.insert(u)$\;
    \For{$s \in E(u)$}{
    	\If{$\big(s \not\in visited$ or $dist[v][s] > dist[v][u] + \omega(u,s)\big)$ and $\big(dist[v][u] + \omega(u,s) \leq k\big)$}{
        	$dist[v][s] = dist[v][u] + \omega(u,s)$\;
            $Q.enqueue(\langle s,dist[v][s]\rangle)$\;
        }
    }
  }
  \For{$u \in visited$}{
  	$\mathit{distances}[v] \gets \mathit{distances}[v] \cup \langle u, \mathit{dist}[u,v] \rangle$\;
  }
 	$\text{sort\_in\_descending\_order}(\mathit{distances}[v])$\;
}

  \caption{Precomputations.}
  \label{algorithm:semi:links:precomputations}
 \end{algorithm}

   \begin{algorithm}[t]

\SetAlgoVlined
\LinesNumbered
 	\For{$u \in V$}{

		$prev\_d  \gets \text{ largest distance in } \mathit{distances}[u]$\;
		$\mathit{Nod_>}[\mathit{prev\_d}] \gets |V| - \mathit{distances}[u]\mathit{.size}()$\;
		$Nod_\geq[prev\_d] \gets |V| - distances[u].size()$\;
        \For{$(v,d) \in distances[u]$}{
        	\If{$d \neq prev\_d$}{
            	$Nod_>[d] \gets Nod_\geq[prev\_d]$\;
                $Nod_\geq[d] \gets Nod_\geq[prev\_d]$\;
                $prev\_d \gets d$\;
            }
            $Nod_\geq[d] \gets Nod_\geq[d] + 1$\;
        }
        $prev\_d  \gets \text{ largest distance in } \mathit{distances}[u]$\;
        
        \For{$k \in [0,|V|]$}{
        $\mathit{MC^-} \gets 0$\;
        \For{$(s,d) \in distances[u]$}{
        	\If{$d \neq prev\_d$}{
                $\mathit{MC^-} \gets \mathit{MC}^- + f(\mathit{prev\_d}) \big( \binom{Nod_{\geq}[d] }{k} \binom{Nod_{>}[d] }{k} \big)$\;
                $prev\_d \gets d$\;
            }
            \For{$(t,dt) \in distances[u]$}{
            	$d \gets \max(d,dt)$\;
            	$\mathit{MC}^+ \gets f(d) \binom{Nod_{\geq}[d] }{k}$\;
                $II[s,t] = MC^- - MC^+$\;
            }
        }
        }

 	}

 \caption{Semivalue Closeness Interaction Index.}
 \label{algorithm:semi:links}
\end{algorithm}

Since our Algorithms require information about the distances between certain nodes to be presented in a sorted fashion and this information is used multiple times, it is advisable to compute these sorted vectors in a precomputation phase. Furthermore, since these precomputations are common between both the Shapley value and semivalue interaction indices, we present them in a common precomputation algorithm, the pseudo-code of which can be found in  Algorithm \ref{algorithm:semi:links:precomputations}. 
Algorithms \ref{algorithm:semi:links} and \ref{algorithm:Shapley:links} continue where the precomputations leave of to compute their respective indices.

In particular, Algorithm \ref{algorithm:semi:links:precomputations} uses a modified Dijkstra's algorithm \cite{dijkstra1959note} in order to compute the distance between the $k$ nearest nodes to each node in $O(|V| (V_k \log(V_k) + E_k))$ time, where $V_k$ is the average number of nodes at distance $k$ from any node and $E_k$ is the average number of edges within a distance of $k$ around any node. 
The algorithm also sets up the data structures required for the computation of the interaction index. 
Algorithms\ref{algorithm:semi:links} and \ref{algorithm:Shapley:links} use dynamic programming in order to compute the negative part of the marginal contributions ($MC^-$ from our proof). Algorithm \ref{algorithm:semi:links} runs in $O(|V|^2 V_k^2)$ time, and Algorithm \ref{algorithm:Shapley:links} runs in $O(|V| V_k^2 + |V|^2)$ time.

We already mentioned that generalised closeness centrality is equivalent to $k$-degree centrality given the appropriate function $f$, and our algorithms can therefore also compute the $k$-degree interaction index. Interestingly, despite being more sophisticated, our algorithm is actually faster.

Szczepanski et al.~\cite{Szczepanski:2015b} quote the complexity of their algorithm as $O(|V|^3)$, however the  authors did not consider the complexity of finding the intersection of two sets. We give the authors the benefit of the doubt, since it is possible to do this in linear time, which gives their algorithm a complexity of $O(|V|^2 V_k)$. This, however, requires a modification of their algorithm in order to sort the neighbour sets in the precomputation phase. 

As opposed to querying every pair of nodes and then every node $s$ within the radius $k$ of the pair, our algorithm reverses this, and first considers any node $s$ and then all pairs of nodes within its vicinity. In doing this, we avoid altogether querying pairs of nodes that are far away (except to first initialise the distance between each node to infinity and interaction index to $0$). Combined with our restricted Dijkstra algorithm, this results in a total time complexity of $O(|V| V_k^2 + |V|^2)$.

  \begin{algorithm}[t]
\SetAlgoVlined
\LinesNumbered
 	\For{$u \in V$}{
		$prev\_d  \gets \text{ largest distance in } \mathit{distances}[u]$\;
		$\mathit{Nod_>}[\mathit{prev\_d}] \gets |V| - \mathit{distances}[u]\mathit{.size}()$\;
		$Nod_\geq[prev\_d] \gets |V| - distances[u].size()$\;
        \For{$(v,d) \in distances[u]$}{
        	\If{$d \neq prev\_d$}{
            	$Nod_>[d] \gets Nod_\geq[prev\_d]$\;
                $Nod_\geq[d] \gets Nod_\geq[prev\_d]$\;
                $prev\_d \gets d$\;
            }
            $Nod_\geq[d] \gets Nod_\geq[d] + 1$\;
        }
        $prev\_d  \gets \text{ largest distance in } \mathit{distances}[u]$\;
        $\mathit{MC^-} \gets 0$\;
        \For{$(s,d) \in distances[u]$}{
        	\If{$d \neq prev\_d$}{
            	$\mathit{MC^-} \gets \mathit{MC}^- + \frac{f(\mathit{prev\_d})}{Nod_<[prev\_d] - 1} - \frac{f(\mathit{prev\_d})}{Nod_\leq[prev\_d] - 1}$\;
                $prev\_d \gets d$\;
            }
            \For{$(t,dt) \in distances[u]$}{
            	$d \gets \max(d,dt)$\;
            	$\mathit{MC}^+ \gets \frac{f(d)}{Nod_\leq[d] - 1}$\;
                $II[s,t] = MC^- - MC^+$\;
            }
        }
 	}
 \caption{Shapley Value Closeness Interaction Index.}
\label{algorithm:Shapley:links}
\end{algorithm}

We study the running time of both algorithms using randomly generated graphs according to the preferential attachment (PA) model due to Barabasi and Albert~\cite{Barabasi:Albert:1999}. In particular, we study two cases: a relatively sparse network, where we start with a clique of $3$ nodes and in each iteration add a node with $2$ edges, and a denser, more centralised network that starts with a clique of $5$ nodes and with each node adds $3$ edges. The running times of both algorithms are presented in Tables \ref{table:interaction_index_running_time:2} and \ref{table:interaction_index_running_time:3}, where $m_0$ is the size of the initial clique and $m$ is the number of edges added during each iteration of the preferential attachment algorithm. We note that the comparison is heavily dependent on our implementation of the algorithms, and that we implemented the algorithm due to Szczepanski et al.~\cite{Szczepanski:2015b} with the benefit of our restricted Dijkstra algorithm. We note that although our algorithm is significantly faster given a sparse network or low radius $k$, due to the more complicated nature of our algorithm it actually becomes somewhat slower given a high enough $V_k$. This is because our algorithm performs more complicated operations, which cannot be expressed by its asymptotic complexity alone.

\begin{table}
{\begin{center}
\begin{tabular}{lllll}
$|V|$ & $k$ & $V_k$ & Algorithm \ref{algorithm:Shapley:links}& Szczepanski et al.~\cite{Szczepanski:2015b}\\\hline
500&1&3.788&3.153&21.789
\\
&2&18.2011&17.417&32.648
\\
&3&65.1542&106.436&107.116
\\
400&1&3.785&2.331&16.412
\\
&2&17.9544&12.295&24.916
\\
&3&57.734&70.226&72.582
\\
300&1&3.78&1.669&10.096
\\
&2&17.1246&9.199&15.626
\\
&3&60.0592&50.777&49.086
\\
200&1&3.77&1.053&6.169
\\
&2&15.5977&5.517&8.668
\\
&3&46.0083&26.046&25.923
\\
100&1&3.74&0.476&2.874
\\
&2&12.0278&2.324&3.416
\\
&3&30.4716&8.147&8.757
\end{tabular}
\end{center}
}
\caption{Average running time (in milliseconds) of Algorithm \ref{algorithm:Shapley:links} compared to Szczepanski et al.~\cite{Szczepanski:2015b} for $1000$ random PA graphs using the parameters $m_0=3$ and $m=2$.}
\label{table:interaction_index_running_time:2}
\end{table}

\section{Empirical Evaluation}

In this section, we compare our generalised closeness Shapley interaction index to four other state-of-the-art link prediction methods from the literature on $11$ real-life networks. 

\subsection{Setting \& Datasets}
We compare our algorithm (referred to as \textbf{Shp. Cls.}) to the Shapley $k$-degree interaction index \cite{Szczepanski:2015b}, to $k$-Common Neighbours \cite{Szczepanski:2015b}, and to the SRW and LRW algorithms \cite{Liu:2010}.
We briefly introduce these algorithms below:


\begin{itemize}
\item \textbf{Shapley $k$-Degree Interaction Index (Shp. Deg.):} This similarity measure is equivalent to our measure with the parameter $f(d) = 1$ for $d \leq k$ and $f(d) = 0$ otherwise.
\item \textbf{$k$-Common Neighbours (CN):} According to this measure, which is used to rank undirected, unweighted graphs, the rank of every non-existing edge between a pair of nodes is the number of common $k$-neighbours between the two nodes. Let $E^k(v) = \set{u \midd v \neq u \text{ and } \dist(v,u) < k}$. $k$-Common Neighbours, then, is defined as follows:
\[\mathit{CN}^k(u,v) = |E^k(u) \cap E^k(v)|\].
\item \textbf{Local Random Walk (LRW):} This measure ranks the similarity of nodes based on the concept of a random walk. Assume that at time step $t = 0$ a walker starts at node $u$. In other words, there is $100\%$ probability that the current node is $u$ at time step $t = 0$. At any other time step, the walker can visit any of the neighbours of the current node with equal probability. Denote by $P_{uv}(t)$ the probability that a walker that started at $u$ is at node $v$ at time step $t$. LRW, then, is defined as follows:
\[\mathit{LRW}^k(u,v) = \frac{|E(u)|}{2|E|} P_{uv}(k) + \frac{|E(v)|}{2|E|} P_{vu}(k)\]
\item \textbf{Superimposed Random Walk (SRW):} This measure is considered by the authors to be a more advanced version of LRW. It is defined as the sum of all LRW measures from time step $0$ to time step $k$. Formally, we have:
\[\mathit{SRW}^k(u,v) = \sum_{t = 0}^k \mathit{LRW}^t(u,v)\]
\end{itemize}

\begin{table}
{\begin{center}
\begin{tabular}{lllll}
$|V|$ & $k$ & $V_k$ & Algorithm \ref{algorithm:Shapley:links}& Szczepanski et al.~\cite{Szczepanski:2015b}\\\hline
500&1&5.184&4.14&39.241
\\
&2&38.5843&39.503&66.001
\\
&3&158.143&360.747&318.472
\\
400&1&5.18&3.437&25.969
\\
&2&34.162&27.707&45.686
\\
&3&139.181&251.65&206.952
\\
300&1&5.1733&2.329&18.799
\\
&2&33.2759&21.42&31.384
\\
&3&133.982&160.85&140.103
\\
200&1&5.16&1.397&10.756
\\
&2&26.6135&11.172&17.757
\\
&3&88.1525&77.701&75.709
\\
100&1&5.12&0.691&5.189
\\
&2&22.4189&4.253&6.685
\\
&3&65.1413&16.191&16.797
\end{tabular}
\end{center}}
\caption{Average running time (in milliseconds) of Algorithm \ref{algorithm:Shapley:links} compared to Szczepanski et al.~\cite{Szczepanski:2015b} for $1000$ random PA graphs using the parameters $m_0=5$ and $m=3$.}
\label{table:interaction_index_running_time:3}
\end{table}


An important facet of our analysis is that whereas, as far as we are aware, the analysis of quasi-local similarity measures in the literature has focused on the performance of algorithms given an optimal choice of $k$, there is no analysis on the impact of a suboptimal $k$ on the algorithms. To combat this, we chose to compare all algorithms using $k$ values of $1$, $2$ and $3$. For the datasets we evaluated, we found that none of the methods significantly benefited from a higher $k$ value (in fact, in most cases a higher value was detrimental), but comparing these $3$ values was sufficient to highlight the differences between the methods.

In order to compare the methods, we take $11$ networks and randomly remove $30\%$ of the edges from each of them. Next, we rank non-existing edges within the networks (including those that were removed) according to each algorithm. We use the area under the curve (AUC) and precision to compare the results. We repeat this process $1000$ times for all networks, methods and parameters $k$ and present the average AUC and precision in Tables \ref{table:shapley:AUC}, and \ref{table:shapley:precision}, respectively. The networks that we used to evaluate the algorithms on\footnote{{\scriptsize{All datasets except Polbbokos are available at \texttt{http://konect.uni-koblenz.de/networks/} \cite{konect:2017}, \texttt{https://snap.stanford.edu/data/} \cite{snapnets}, or \texttt{http://vlado.fmf.uni-lj.si/pub/networks/data/}  \cite{pajek}} }} are as follows: 
Youtube 20 and Amazon 100 \cite{snapnets}, Football \cite{Girvan:2002,pajek}, Taro \cite{konect:hage,konect:schwimmer,konect:2017}, Jazz \cite{konect:arenas-jazz,konect:2017}, Zachary \cite{zachary:1977,konect:2017}, and Dolpins cite{konect:dolphins,konect:2017}.


\begin{itemize}
\item \textbf{Youtube 20:} A network of users belonging to the $20$ top groups in the SNAP dataset from the popular video-sharing website Youtube \cite{snapnets}. Connections between users indicate friendship between their user accounts. Link prediction can predict friendships between users whose user accounts are not formally connected as friends on the website.
\item \textbf{Amazon 100:} A network of products from the $100$ top product categories in the SNAP dataset from the Amazon online store website. Connections between products were mined using the ``Customers Who Bought This Item Also Bought'' feature \cite{snapnets}. Link prediction methods can be used to discover new product associations and therefore improve the impact of the recommendation service.
\item \textbf{US AIR:} A network of airports in the USA and their connections \cite{pajek}. Link prediction can be used as a method of predicting up and coming flight connections.
\item \textbf{Football:} A network of college football teams \cite{Girvan:2002,pajek}. Edges represent matches between the teams.
\item \textbf{Taro:} A network of gift-giving (taro) between households in a Papuan village \cite{konect:hage,konect:schwimmer,konect:2017}.
\item \textbf{Jazz:} A collaboration network of jazz musicians from 2003 \cite{konect:arenas-jazz,konect:2017}. Edges indicate that two musicians performed together in a band.
\item \textbf{Zachary:} A friendship network of the Zachary karate club \cite{zachary:1977,konect:2017}.
\item \textbf{Surfers:} A network of the interpersonal contacts of windsurfers in southern California in the fall of 1986 \cite{konect:freeman1988,konect:2017}.
\item \textbf{Dolphins:} A network representing a community bottlenose dolphins off Doubtful Sound and their associations observed between 1994 and 2001 \cite{konect:dolphins,konect:2017}.
\item \textbf{Terrorists:} The network of suspected terrorists who orchestrated the 2004 Madrid train bombing \cite{konect:hayes,konect:2017}. A connection between two terrorists indicates that they communicated with each other.
\item \textbf{Polbooks:} This network represents books about politics sold through Amazon. Two books are connected if they were frequently co-purchased \cite{wu:2013}. This dataset was kindly provided by \cite{Szczepanski:2016b}.
\end{itemize}

We present some of the characteristics of the networks in Table~\ref{table:network:statistics}.

\begin{table}

\begin{center}
{\small{
\begin{tabular}{lll|l|r}
Dataset&$|V|$&$|E|$&$k$&Average $V_k$ with $30\%$ \\&&&& edges randomly removed
\\\hline
Youtube 20&436&1384&1&3.22018
\\
&&&2&16.8829
\\
&&&3&40.7019
\\
Amazon 100&433&2014&1&4.25173
\\
&&&2&6.39215
\\
&&&3&6.81899
\\
Football&115&1226&1&8.46087
\\
&&&2&34.3531
\\
&&&3&87.2797
\\
 Taro&22&78&1&3.45455
 \\
 &&&2&7.18136
 \\
 &&&3&11.602
 \\
Jazz&198&5484&1&20.3838
\\
&&&2&113.99
\\
&&&3&177.578
\\
Zachary&32&156&1&4.27647
\\
&&&2&14.1438
\\
&&&3&22.6532
\\
Surfers&43&672&1&11.9302
\\
&&&2&38.1082
\\
&&&3&42.939
\\
Dolphins&62&318&1&4.58065
\\
&&&2&13.7555
\\
&&&3&26.2656
\\
Polbooks&105&822&1&6.86667
\\
&&&2&28.5871
\\
&&&3&54.9784
\\

\end{tabular}
}}
\end{center}
\caption{The networks' characteristics.}
\label{table:network:statistics}
\end{table}

\subsection{Results}

\begin{table}
{\small{
\begin{center}

\begin{tabular}{ll|lllll}
Dataset&$k$&Shp. Cls.&Shp. Deg.&CN&SRW&LRW
\\\hline
Youtube 20&1&61.351&61.351&61.102&58.75&58.75
\\
&2&\colorbox{gray!30!}{70.016}&69.723&69.175&63.561&63.626
\\
&3&66.896&66.527&65.898&62.391&62.613
\\
Amazon 100&1&92.554&92.554&92.466&74.426&74.426
\\
&2&96.961&96.914&96.875&91.641&91.594
\\
&3&\colorbox{gray!30!}{96.992}&96.99&96.834&92.031&91.937
\\
 US Air&1&\colorbox{gray!30!}{92.94}&\colorbox{gray!30!}{92.94}&91.654&86.436&86.436
 \\
 &2&91.591&88.693&85.598&91.695&91.801
 \\
 &3&91.285&84.295&83.885&91.420&90.827
 \\
Football&1&81.361&81.361&81.382&67.18&67.18
\\
&2&\colorbox{gray!30!}{82.861}&80.392&77.99&77.559&78.423
\\
&3&81.291&54.998&53.042&76.832&74.938
\\
 Taro&1&\colorbox{gray!30!}{59.282}&\colorbox{gray!30!}{59.282}&59.087&48.41&48.41
 \\
 &2&51.26&49.354&42.494&48.662&49.17
 \\
 &3&48.993&44.516&40.734&45.985&44.04
 \\
Jazz&1&\colorbox{gray!30!}{95.836}&\colorbox{gray!30!}{95.836}&94.412&84.561&84.561
\\
&2&94.497&84.4259&80.361&90.442&90.73
\\
&3&94.995&74.376&72.735&89.434&85.352
\\
Zachary&1&65.883&65.883&63.412&54.199&54.199
\\
&2&67.76&63.607&59.996&66.065&67.447
\\
&3&67.842&62.803&60.997&67.078&\colorbox{gray!30!}{68.351}
\\
Surfers&1&\colorbox{gray!30!}{82.091}&\colorbox{gray!30!}{82.091}&80.668&58.715&58.715
\\
&2&81.362&69.836&68.69&66.917&69.018
\\
&3&81.825&52.137&52.131&65.899&63.69
\\
Dolphins&1&71.397&71.397&71.493&62.996&62.996
\\
&2&\colorbox{gray!30!}{77.142}&76.716&74.67&72.7849&73.160
\\
&3&76.816&74.841&71.378&72.618&72.235
\\
 Terrorists&1&\colorbox{gray!30!}{89.573}&\colorbox{gray!30!}{89.573}&88.173&68.404&68.404
 \\
 &2&88.992&85.685&79.728&82.872&84.334
 \\
 &3&88.469&78.282&75.632&83.128&82.86
 \\
Polbooks&1&83.795&83.795&83.0612&73.511&73.511
\\
&2&\colorbox{gray!30!}{87.898}&85.828&83.379&83.629&84.643
\\
&3&87.515&81.134&79.551&83.741&83.301
\\

\end{tabular}

\end{center}
}}
\caption{Average AUC (best result indicated in gray).}
\label{table:shapley:AUC}
\end{table}

\begin{table}
{\small{
\begin{center}

\begin{tabular}{ll|lllll}
Data Set&$k$&Shp. Cls.&Shp. Deg.&CN&SRW&LRW
\\\hline
Youtube 20&1&9.23527&9.23527&6.11643&4.17198&4.17198
\\
&2&8.73816&3.22899&1.35217&5.58599&7.39179
\\
&3&\colorbox{gray!30!}{9.37343}&5.70435&4.96667&6.82899&7.91884
\\
Amazon 100&1&60.594&60.594&58.1&49.129&49.129
\\
&2&\colorbox{gray!30!}{60.988}&60.587&49.064&57.1358&58.59
\\
&3&60.865&64.735&36.161&57.443&58.256
\\
 US Air&1&\colorbox{gray!30!}{53.8612}&\colorbox{gray!30!}{53.8612}&45.1463&41.0411&41.0411
 \\
 &2&49.186&37.0044&34.402&47.1418&48.3915
 \\
 &3&50.9319&27.8438&28.8824&47.0859&44.7281
 \\
Football&1&40.856&40.856&\colorbox{gray!30!}{41.105}&21.958&21.958
\\
&2&40.546&20.489&17.596&28.272&32.802
\\
&3&40.104&3.3776&3.1885&28.52&26.277
\\
 Taro&1&\colorbox{gray!30!}{15.9455}&\colorbox{gray!30!}{15.9455}&12.9455&14.1636&14.1636
 \\
 &2&15.2364&11.4091&3.8664&13.7182&10.7909
 \\
 &3&12.9273&4.00909&3.12727&12.4273&7.73636
 \\
Jazz&1&\colorbox{gray!30!}{63.697}&\colorbox{gray!30!}{63.697}&57.666&36.83&36.83
\\
&2&60.812&28.051&24.832&42.33&42.244
\\
&3&62.261&19.901&18.588&40.858&36.682
\\
Zachary&1&\colorbox{gray!30!}{24.487}&\colorbox{gray!30!}{24.487}&15.544&10.5652&10.565
\\
&2&19.387&9.0522&7.5696&15.178&20.096
\\
&3&20.587&13.974&13.778&17.704&19.913
\\
Surfers&1&49.378&49.378&47.155&26.213&26.213
\\
&2&49.287&33.981&33.552&30.671&33.475
\\
&3&\colorbox{gray!30!}{49.391}&25.035&25.052&30.424&28.723
\\
Dolphins&1&18.447&18.447&20.753&13.879&13.879
\\
&2&\colorbox{gray!30!}{20.232}&18.002&12.172&15.975&16.515
\\
&3&19.728&11.355&9.7085&14.977&15.496
\\
 Terrorists&1&\colorbox{gray!30!}{65.9319}&\colorbox{gray!30!}{65.9319}&55.8806&26.1722&26.1722
 \\
 &2&64.4972&41.8264&35.7958&37.4875&45.5889
 \\
 &3&64.8903&28.5861&28.5806&39.3431&39.0875
 \\
Polbooks&1&29.367&29.367&26.664&16.884&16.884
\\
&2&29.739&22.363&21.186&20.049&22.404
\\
&3&\colorbox{gray!30!}{29.899}&15.115&13.077&20.767&20.427
\end{tabular}

\end{center}
}}
\caption{Average precision (best result indicated in gray).
}
\label{table:shapley:precision}
\end{table}


In all but the Zachary network (where $\mathit{LRW}^3$ achieved the best result), our closeness Shapley interaction index achieved the best AUC, and in all but the Football network (where the simplest common neighbour algorithm with a radius of $1$ achieved the best result) our algorithm achieved the best precision. Even in these two networks, however, the comparative advantage of other methods was marginal. Interestingly, in the Football and Dolphins networks the quality the AUC of our method increased when $k$ was raised from $1$ to $2$, however the quality of the Shapley $k$-degree interaction index fell. We attribute this to the fact that the latter algorithm over-stresses the importance of second- and third-order relationships. For precision, we see this phenomenon in the Amazon 100, Dolphins, and most prominently Surfers and Polbooks (where our ranking was slightly improved, but the Shapley degree ranking decreased by approximately $24\%$ and $14\%$, respectively) networks.

In general, we note that the quality of all algorithms tends to fall when the $k$ value is too high. It seems that the Shapley value closeness measure is very resilient to this phenomenon. Even in cases when the quality of SRW and LRW does not fall as much, they produce worse results at any $k$ value, making this irrelevant. In fact, the results of both algorithms given any $k$ were worse than random in the Taro network. This network seems especially difficult, with many of the results being worse than random. We also note that our algorithm is---generally---more resilient than LRW and SRW when choosing a $k$ value that is too low.

We see that the Shapley degree interaction index and common neighbours are the most likely to under-perform given a high $k$ value. In fact, this can decrease the AUC of both measures by nearly $30\%$ (as seen in the Surfers network), and precision by approximately $40\%$ (in the case of the Jazz and Football networks). Given this, it is difficult to recommend these algorithms as quasi-local link prediction methods, given that providing them with too much information (i.e., a $k$ parameter that is too high) can result in a ranking that is little better than random. Although it is possible to estimate this parameter when using these methods, there is no way to know whether the parameter is too high, potentially dramatically decreasing the effectiveness of the measure.

Finally, we note that whereas SRW is generally viewed as being superior to LRW, it is LRW that usually achieves the better result in our experiments when each algorithm is given its own, respective optimal paramter $k$. We note, however, that given the same $k$ for a high value of $k$, it is usually SRW that is superior.

\section{Conclusions and Future Work}
We developed a new game-theoretic quasi-local algorithm for link prediction that is based on generalised closeness centrality. Our approach achieves competitive results when compared to the state-of-the-art, especially given a suboptimal radius, $k$, within which to query for similarities between nodes. 

We are particularly keen on two future research directions. First, we aim to study a variable radius for different pairs of nodes. It goes to reason that if a different radius is required to achieve the optimal result in different networks, perhaps various sections of the network (such as connected components) should also be studied with a different radii. Moreover, other group centrality measures may prove to be even more effective when paired with the interaction index in predicting links between nodes. Group betweenness centrality \cite{Everett:Borgatti:1999,Szczepanski:2016}, for example, has not been studied for this purpose.

We are also keen on studying how resilient the game-theoretic link prediction algorithm proposed in this paper is to strategic manipulation by an evader who purposefully attempts to hide her links. A number of such studies have been recently proposed in the literature~\cite{waniek2019hide,zhou2019adversarial,zhou2019attacking,chen2018link}. Interestingly, in a similar line of research on evading detection by centrality measures~\cite{waniek2018hiding,waniek2017construction}, it has been shown that game-theoretic centrality measures~\cite{tarkowski2018efficient,skibski2018axiomatic,Michalak:et:al:2015} are more difficult to evade than the conventional ones~\cite{Baranowski:Gorski:2018}. We believe the same will be the case

\section{Acknowledgements}
This paper was supported by the Polish National Science Centre grant DEC-2013/09/D/ST6/03920. Michael Wooldridge was supported by the European Research Council under Advanced Grant 291528 (``RACE'').
\bibliographystyle{abbrv}
\bibliography{referencesconsolidated}

\end{document}